\documentclass[orivec,envcountsame]{llncs}
\usepackage{amssymb}
\usepackage{amsmath}
\title{Short proofs of strong normalization}
\author{Aleksander Wojdyga}
\institute{Faculty of Mathematics and Computer Science \\Nicolaus Copernicus University \\ Toru\'n \\
\email{awojdyga@mat.uni.torun.pl}
\and Institute of Computer Science, \\ Faculty of Electrical Engineering and Computer Science, \\ Lublin University of Technology, \\ Lublin}

\newcommand{\strzjn}[1]{\ensuremath{{#1}_1 \rightarrow \cdots \rightarrow {#1}_n}} % od jeden do n
\newcommand{\strzjx}[2]{\ensuremath{{#1}_1 \rightarrow \cdots \rightarrow {#1}_{#2}}} % od jeden do x
\newcommand{\para}[1]{\ensuremath{\langle #1 \rangle}}

\newcommand{\ra}{\rightarrow}
\newcommand{\thra}{\twoheadrightarrow}
\newcommand{\tr}[1]{{\ensuremath{\left \vert #1 \right \vert}}}

\newcommand{\trtyp}[2]{{\ensuremath{{\left \vert #1 \right \vert}^{\left \vert #2 \right \vert}}}}

\newcommand{\inl}{\ensuremath{\mathtt{in}_1}}
\newcommand{\inr}{\ensuremath{\mathtt{in}_2}}
\newcommand{\prl}{\ensuremath{\pi_1}}
\newcommand{\prr}{\ensuremath{\pi_2}}

\newcommand{\lambdaiakf}{\lambda_{\ra,\wedge,\vee,\bot}} % lambda + implikacja alternatywa koniunkcja falsz
\newcommand{\Fi}{\mathbf{F}_{\forall,\ra}}
\newcommand{\Fiakf}{\mathbf{F}_{\forall,\exists,\ra,\wedge,\vee,\bot}} % F + uni egz implikacja alternatywa koniunkcja falsz
\newcommand{\cud}[1]{\ensuremath{\epsilon_#1}}

\newcommand{\rabeta}{\rightarrow_\beta}
\newcommand{\cmt}{\rightsquigarrow}
\newcommand{\Ra}{\diamond}
\newcommand{\at}{\,@\,}
\newcommand{\trf}[1]{{\underline{#1}}}
\newcommand{\rhs}{\mathrm{RHS}}
\newcommand{\lhs}{\mathrm{LHS}}
\newcommand{\raa}{\twoheadrightarrow}

\begin{document}

\maketitle

\begin{abstract}
This paper presents simple, syntactic strong normalization proofs for the simply-typed 
$\lambda$-calculus and the polymorphic $\lambda$-calculus (system {\bfseries F}) with 
the full set of logical connectives, and all the  permutative reductions. The normalization 
proofs use translations of terms and types of $\lambdaiakf$ to terms and types of $\lambda_\ra$ 
and from $\Fiakf$ to $\Fi$.
\end{abstract}

\section{Introduction}

In this paper we consider the simply-typed and polymorphic lambda-calculus 
extended by type constructors corresponding to the usual logical
connectives,
namely conjunction, disjunction, absurdity and implication. In the
polymorphic
case we include both universal and existential quantification. In
addition, we
assume all the permutative conversions. 

Different proofs of strong normalization of several variants of these
calculi 
occur in the literature~cf.~\cite{davr03,joafm03,schh05,tatm05,tatm07}. 
It is however surprising that it is 
quite hard to find one covering the full set of connectives, applying to 
all the permutative conversions (in the polymorphic case none of the cited 
works does so) and given by a simple and straightforward argument. We can 
only repeat after J.Y. Girard: {\it I didn't find a proof really nice,
and taking little space}~\cite[p.~130]{blindspot}. 
For instance, many proofs, like these in~\cite{schh05,tatm05,tatm07} 
are based on the 
computability method, or (in the polymorphic case) candidates of
reducibility.
This requires re-doing each time the same argument, but in a more complex 
way, due to the increased complexity of the language. 

We believe that methodologically the most adequate approach is by
reducing 
the question of strong normalization of the extended systems to the known 
strong normalization of the base systems, involving only implication and
the 
universal quantifier. We propose two such proofs in what follows.

The first proof reduces the calculus $\lambdaiakf$ with connectives $\wedge,\vee,\ra,\bot$ to the calculus $\lambda_\ra$. 
Here we use the strong normalization of $\lambda_\ra$ with beta-eta-reductions. 
The proof is based on composing the ordinary reduction of classical
connectives to implication and absurdity with Ong's translation 
of the $\lambda\mu$-calculus to the ordinary $\lambda\eta$-calculus,
as described e.g.~in~\cite[Chapter 6]{sorm2006}. To our knowledge 
this is the most direct way of showing SN for system $\lambdaiakf$.

The above method does not however extend to the polymorphic case.
Indeed, the translation is strictly type-driven and requires
an \emph{a priori} knowledge of all types a given expression can obtain
by polymorphic instantiation. Also the well known 
definition of logical connectives in system {\bfseries F}:
\begin{displaymath}
 \sigma \wedge \tau \equiv \mathrm\forall t . (\sigma \ra \tau \ra t) \ra t \qquad
 \sigma \vee \tau \equiv \mathrm\forall t . (\sigma \ra t) \ra (\tau \ra t) \ra t
\end{displaymath}
is not adequate. The translation preserves beta-conversion, but
not the permutations. The solution, first used by de Groote~(\cite{groote99}, \cite{groote02}),
for first-order logic,
is a~CPS-translation. Our proof is similar to de Groote's but the 
version of CPS we use is based on Nakazawa and Tatsuta~\cite{NakazawaTatsuta08}.

%The full set of logical connectives is important in lambda calculus and intuitionistic logic {\tt [sens?]} . In the latter system the connectives cannot be expressed one with another{\tt [i co z tego?]}. It is also important for computer science, where conjunction and disjunction allow to encode advanced data types easily{\tt [mianowicie?]}

%Strong normalization property of those calculi is not only a important field of research {\tt [field of research?!]} . The normal proofs in intuitionistic logic enjoy subformula property, ie. any formula in a normal proof is a subformula of the end sequent{\tt [W jakim sensie to zdanie pozostaje w opozycji do poprzedniego?]}. This transfers \emph{via} Curry-Howard isomorphism to normal terms of $\lambda$ calculus. The term in a normal form clearly shows,{\tt [jak?]} what the program (term) does and how it achieves the effect of its computation.

\subsection{Definitions of relevant calculi}

We consider the calculi $\lambdaiakf$ and $\Fiakf$ in Church's style.
The type $\tau$ of a term $M$ is written informally in upper index as $M^\tau$. However, if it is clear from the context, types will be omitted for the sake of brevity and readability -- most right-hand sides of equations and reduction rules are written without types.

\subsubsection{The full simply-typed $\lambda$-calculus}
Types of $\lambdaiakf$ are built from multiple type constants; lowercase Greek letters are used to denote types.
\begin{definition}\rm \label{typys} Types of $\lambdaiakf$ 
 \begin{equation*}
  \sigma, \tau, \ldots ::= p, q, \ldots, \sigma \ra \tau, \sigma \wedge \tau, \sigma \vee \tau, \bot
 \end{equation*}
\end{definition}

\noindent
Syntax of terms of $\lambdaiakf$ can be divided in two groups: constructor terms and eliminator terms. Lowercase Latin letters denote variables, uppercase -- terms.
\begin{definition}\rm \label{termys} Terms of $\lambdaiakf$
 \begin{eqnarray*}
  M, N, \ldots &::=& \mbox{\em Variables} \nonumber\\
  & & x^\sigma, y^\tau, \ldots, \nonumber\\
  & & \mbox{\em Introduction} \nonumber\\
  & & (\lambda x^\sigma. N ^\tau)^{\sigma\ra\tau}, \para{M^\sigma, N^\tau}^{\sigma\wedge\tau}, (\inl{A^\sigma})^{\sigma\vee\tau}, (\inr{B^\tau})^{\sigma\vee\tau} \nonumber\\
  & & \mbox{\em Elimination} \nonumber\\
  & & (M^{\sigma\ra\tau} N^\sigma)^\tau, (P^{\sigma\wedge\tau}\prl)^\sigma, (P^{\sigma\wedge\tau}\prr)^\tau, (W^{\sigma\vee\tau}[x^\sigma.S^\delta,y^\tau.T^\delta])^\delta, \nonumber\\
  & & (A^\bot \cud{\tau})^\tau
 \end{eqnarray*}
\end{definition}

\noindent
In the above, the notation $\inl{A}$ and $\inr{A}$ represents the left and
right injection for the sum type, $\pi_1$ and $\pi_2$ are projections
and $W ^{\sigma\vee\tau}[x.S^\delta,y.T^\delta]$ stands 
for a~case statement. The epsilon represents the {\it ex falso\/}.
%% $\caseofor{W^{\sigma\vee\tau}}{[x^\sigma] S^\delta}{[y^\tau] T^\delta}$.

\subsubsection{Reductions}
The beta-reductions are written as $\rabeta$ and commutative reductions are denoted by $\rightsquigarrow$.
For any reduction $\ra$ transitive closure of this relation will be denoted as $\ra^+$ and transitive, reflexive closure as $\raa$.

\begin{definition}\rm \label{betas} $\beta$-reductions in $\lambdaiakf$
 \begin{align*}
  (\lambda x^\tau . M^\delta) A^\tau &\rabeta M[x:=A]^\delta \\
  \para{M^\sigma, N^\tau} \prl &\rabeta M^\sigma \\
  \para{M^\sigma, N^\tau} \prr &\rabeta N^\tau \\
  (\inl A)^{\sigma\vee\tau} [x^\sigma.S^\delta,y^\tau.T^\delta] &\rabeta S[x^\sigma:=A^\sigma]^\delta \\
  (\inr B)^{\sigma\vee\tau} [x^\sigma.S^\delta,y^\tau.T^\delta] &\rabeta S[y^\tau:=B^\tau]^\delta
 \end{align*}
\end{definition}

\begin{definition}\rm \label{cmts} Commutative reductions in $\lambdaiakf$
 \begin{align*}
  (A^\bot \cud{{\sigma\ra\tau}}) N^\sigma & \cmt A^\bot \cud{\tau} \\
  (A^\bot \cud{{\sigma\wedge\tau}}) \prl & \cmt A^\bot \cud{\sigma} \\
  (A^\bot \cud{{\sigma\wedge\tau}}) \prr & \cmt A^\bot \cud{\tau} \\
  (A^\bot \cud{{\sigma\vee\tau}}) [x^\sigma.S^\delta, y^\tau.T^\delta] & \cmt A^\bot \cud{\delta} \\
  (A^\bot \cud{\bot}) \cud{\sigma} & \cmt A^\bot \cud{\sigma} \\
  ((W^{\sigma\vee\tau} [x.S^{\alpha\ra\beta}, y.T^{\alpha\ra\beta}]) N^\alpha)^\beta 
   & \cmt W^{\sigma\vee\tau}[x.(SN)^\beta, y.(TN)^\beta] \label{elimapp} \\
  ((W^{\sigma\vee\tau} [x.S^{\alpha\wedge\beta}, y.T^{\alpha\wedge\beta}]) \prl)^\alpha
   & \cmt W^{\sigma\vee\tau}[x.(S\prl)^\alpha, y.(T\prl)^\alpha] \\
  ((W^{\sigma\vee\tau} [x.S^{\alpha\wedge\beta}, y.T^{\alpha\wedge\beta}]) \prr)^\beta
   & \cmt W^{\sigma\vee\tau}[x.(S\prr)^\beta, y.(T\prr)^\beta] \\
  \begin{split}
   (W^{\sigma\vee\tau} [x.S^{\alpha\vee\beta}, y.T^{\alpha\vee\beta}])
    [a^\alpha.A^\delta, & \, b^\beta.B^\delta] \cmt \\
     W^{\sigma\vee\tau} [x.S & [a.A^\delta, b.B^\delta], y.T[a.A^\delta, b.B^\delta]]
  \end{split} \\
 (W^{\sigma\vee\tau} [x.S^\bot, y.T^\bot]) \cud{\alpha}
  & \cmt W^{\sigma\vee\tau}[x.S\cud{\alpha}, y.T\cud{\alpha}]
\end{align*}
\end{definition}
Note that the above commutative reductions follow these two patterns:
\begin{eqnarray}
(W[x.S,y.T])E &\cmt& W[x.SE,y.TE], \label{cmts-pat1} \\
(A\cud{{}})E &\cmt& A\cud{{}}, \label{cmts-pat2}
\end{eqnarray}
where $E$ is an arbitrary eliminator. 
That is, $E$ is either a term $N$ or a~projection, or epsilon, or it has the
form $[x.S,y.T]$.

\subsubsection{The full polymorphic $\lambda$-calculus}

The full polymorphic $\lambda$-calculus extends the system of the previous section by existential and universal polymorphism. Terms of the calculus are all the terms of simply-typed $\lambda$ calculus plus universal and existential introduction and elimination. 

\begin{definition}\rm \label{typyf} Types of $\Fiakf$
 \begin{equation*}
  \sigma, \tau, \ldots ::= p, q, \ldots, \sigma \ra \tau, \sigma \wedge \tau, \sigma \vee \tau, \forall p \, \tau, \exists p \, \tau, \bot
 \end{equation*}
\end{definition}

In the definition below, notation $[M^{\tau[p := \sigma]}, \sigma]$ stands for introduction of type $\exists p \, \tau$ and $[x^\tau.N^\delta]$ is a eliminator for that type.

\begin{definition}\rm \label{termyf} Terms of $\Fiakf$
 \begin{equation*}
  \begin{split}
   M, N, \ldots ::=& \mbox{\em Variables} \\
    & x^\sigma, y^\tau, \ldots \\
    & \mbox{\em Introductions} \\
    & (\lambda x^\sigma. N ^\tau)^{\sigma\ra\tau}, \para{M^\sigma, N^\tau}^{\sigma\wedge\tau},
      (\inl{A^\sigma})^{\sigma\vee\tau}, (\inr{B^\tau})^{\sigma\vee\tau}, \\
    & [M^{\tau[p := \sigma]},\sigma]^{\exists p \, \tau}, (\mathrm\Lambda p M^\tau)^{\forall p \, \tau} \\
    & \mbox{\em Eliminations} \\
    & (M^{\sigma\ra\tau} N^\sigma)^\tau, (P^{\sigma\wedge\tau}\prl)^\sigma, (P^{\sigma\wedge\tau}\prr)^\tau,
      (W^{\sigma\vee\tau}[x^\sigma.S^\delta,y^\tau.T^\delta])^\delta, \\
    & (M^{\exists p \, \tau}[x^\tau.N^\delta])^\delta, (M^{\forall p \, \tau} \sigma)^{\tau[p := \sigma]} \\
    & (A^\bot \cud{\tau})^\tau
  \end{split}
 \end{equation*}
\end{definition}

\noindent
The~$\beta$-reductions and commutative reductions in this system are as follows.
\begin{definition}\rm
\label{betaf} The $\beta$-reductions in $\Fiakf$ are as in Definition \ref{betas} and in
addition
 \begin{align}
  [M^{\tau[p := \sigma]}, \sigma][x^\tau.N^\delta] & \rabeta (N[p:=\sigma][x := M])^\delta \label{elimegz} \\
  (\mathrm\Lambda p M^\tau) \sigma & \rabeta M{[p := \sigma]}
 \end{align}
\end{definition}

\noindent 
The total number of commutative reductions reaches 21. The patterns mentioned in Rules (\ref{cmts-pat1}) and (\ref{cmts-pat2}) are extended by the additional one:
\begin{align}
 (M[x.P])E \cmt M[x.PE], \label{cmts-pat3}
\end{align}
where $E$ can also be of the form of existential ($[y.R]$) or universal ($\sigma$) eliminator.
\begin{definition}\rm \label{cmtf} Additional commutative reductions in $\Fiakf$. \smallskip
{\allowdisplaybreaks

\noindent 
\mbox{Let $\delta$ abbreviate $\forall p \, \alpha$ in rules below.}
 \begin{align}
  (W^{\sigma\vee\tau} [x^\sigma.S^\delta,y^\tau.T^\delta]) \gamma \cmt
   & W[x.(S\gamma)^{\alpha[p:=\gamma]}, y.(T\gamma)^{\alpha[p:=\gamma]}] \\
  (A^\bot \cud{\delta}) \gamma \cmt & A^\bot \cud{{\alpha[p:=\gamma]}} \\
  (M^{\exists p \, \tau} [x^\tau.P^\delta]) \gamma \cmt
   & M^{\exists p \, \tau} [x.(P\gamma)^{{\alpha[p:=\gamma]}}] \\
 \end{align}
}
\noindent 
 \mbox{In the following rules, $\delta$ abbreviates $\exists p \, \alpha$.}
{\allowdisplaybreaks
 \begin{align}
  (W^{\sigma\vee\tau}[x^\sigma.S^\delta, y^\tau.T^\delta]) [a^\alpha.N^\xi] \cmt
   & W^{\sigma\vee\tau}[x.(S[a.N])^\xi, y.(T[a.N])^\xi] \\
  (A^\bot \cud{\delta}) [a^\alpha.N^\xi] \cmt &  A^\bot \cud{\xi} \\
  (M^{\exists p \, \tau} [y^\tau.P^\delta]) [a^\alpha.N^\xi] \cmt
   & M^{\exists p \, \tau} [y.(P [a . N])^\xi] \label{Eelimelim} \\
  A^\delta[x^\alpha. N ^{\sigma\ra\tau}] P^\sigma \cmt & A[x.(NP)^\tau] \\
  A^\delta[x^\alpha. N ^{\sigma\wedge\tau}] \prl \cmt & A[x . (N\prl)^\sigma] \\
  A^\delta[x^\alpha. N ^{\sigma\wedge\tau}] \prr \cmt & A[x . (N\prr)^\tau] \\
  A^\delta[x^\alpha. N ^{\sigma\vee\tau}] [y^\sigma.S^\delta,z^\tau.T^\delta] \cmt
   & A[x.(N[y.S,z.T])^\delta] \\
  A^\delta[x^\alpha. N ^\bot] \cud{\sigma} \cmt & A[x.(N \cud{\sigma})^\sigma] \label{Aelimcud}
 \end{align}
}
\end{definition}

\section{The translation for simple types}
A type $\tau$ of the $\lambdaiakf$ calculus is translated
 to a type $\tr{\tau}$ of $\lambda_\ra$
calculus, a~term $M$ is translated to a term $\tr{M}$.

\begin{definition}\rm \label{trtyp} Translation of types.
\begin{align*}
\tr{\alpha} &= \bot, \; \mbox{{\textnormal for all type constants} $\alpha = \bot, p, q, \ldots$} \\
\tr{\sigma \ra \tau} &= \tr{\sigma} \ra \tr{\tau} \\
\tr{\sigma \wedge \tau} &= (\tr{\sigma} \ra \tr{\tau} \ra \bot) \ra \bot \\
\tr{\sigma \vee \tau} &= (\tr{\sigma} \ra \bot) \ra (\tr{\tau} \ra \bot) \ra \bot
\end{align*}
\end{definition}

\begin{example}
Let $\tau = p \ra q \ra (p \wedge q)$. Then 

\hfil $\tr{\tau} = \bot \ra \bot \ra (\bot \ra \bot \ra \bot) \ra \bot$.
\end{example}

\begin{definition}\rm \label{trterm} (Translation of terms)
It is assumed below that types $\tr{\sigma}, \tr{\tau}$ and $\tr{\delta}$ are
as follows:
 $\tr{\sigma} = \strzjn{\sigma} \ra \bot$,
 $\tr{\tau} = \strzjx{\tau}{m} \ra \bot$ and
 $\tr{\delta} = \strzjx{\delta}{k} \ra \bot$.
\begin{eqnarray}
\tr{x^\sigma} &=& x^\tr{\sigma} \\
\tr{\lambda x^\tau . M^\sigma} &=& \lambda x^\tr{\tau} . \trtyp{M}{\sigma} \\
\tr{\para{M, N}^{\sigma \wedge \tau}} &=& \lambda z^{\tr{\sigma} \ra \tr{\tau} \ra \bot}.z
\trtyp{M}{\sigma} \trtyp{N}{\tau} \label{para} \\
\tr{{(\inl A)}^{\sigma \vee \tau }} &=& \lambda x^{\tr{\sigma} \ra \bot}.\lambda y^{\tr{\tau} \ra \bot}.x\trtyp{A}{\sigma} \\
\tr{{(\inr B)}^{\sigma \vee \tau}} &=& \lambda x^{\tr{\sigma} \ra \bot}.\lambda y^{\tr{\tau} \ra \bot}.x\trtyp{B}{\tau} \\
\tr{(M^{\sigma \ra \tau} N^{\sigma})} &=& (\tr{M}^{\tr{\sigma} \ra \tr{\tau}} \trtyp{N}{\sigma}) \\
\tr{(P^{\sigma \wedge \tau}) \prl} &=& \lambda x_1^{\sigma_1} \ldots \lambda x_n^{\sigma_n} . 
 \trtyp{P}{\sigma \wedge \tau} \nonumber\\ 
 & & (\lambda x^\tr{\sigma}. \lambda y^\tr{\tau}. (x x_1 \ldots x_n)^\bot) \label{rzut} \\
\tr{(P^{\sigma \wedge \tau}) \prr} &=& \lambda x_1^{\tau_1} \ldots \lambda x_m^{\tau_m} . 
 \trtyp{P}{\sigma \wedge \tau} \nonumber\\ 
 & & (\lambda x^\tr{\sigma}. \lambda y^\tr{\tau} . (y x_1 \ldots x_m)^\bot) \\
\tr{A^{\sigma \vee \tau}[x.S^\delta, y.T^\delta]} &=& \lambda x_1^{\delta_1} \ldots \lambda x_k^{\delta_k}.   %\nonumber\\  & &
 \tr{A}^{(\tr{\sigma} \ra \bot) \ra (\tr{\tau} \ra \bot) \ra \bot} \nonumber\\
 & & (\lambda x^{\tr{\sigma}}.\trtyp{S}{\delta} x_1 \ldots x_k) 
 %\nonumber\\ & & 
 (\lambda y^{\tr{\tau}}.\trtyp{T}{\delta} x_1 \ldots x_k) \label{elimcase} \\
\tr{M^\bot \cud\sigma} &=& \lambda x_1^{\sigma_1} \ldots \lambda x_{n-1}^{\sigma_{n-1}} . \tr{M}^\bot
\end{eqnarray}
\end{definition}

\begin{lemma} [\emph{Soundness}]
\label{trsoundthm} If a term $M$ has type $\delta$, then
$\tr{M}$ has type $\tr{\delta}$.
\end{lemma}

\begin{proof} Obvious.
%% The proof follows by induction on the syntax of term $M$. However, one has to notice that, the terms are in Church's style, therefore it suffices to analyse types of bound and free variables. The theorem is obvious basing on the Definition \ref{termys}.
\qed
\end{proof}

\begin{lemma} \label{betaeta} If $R \ra R'$, then $\tr{R} \ra_{\beta\eta}^{+} \tr{R'}$.
\end{lemma}

\begin{proof}
 The proof proceeds by cases on the definition of $\rabeta$ and $\cmt$. Two example reductions
 will be elaborated here.\smallskip %%

 \noindent(\ref{rzut})
Let $R=\para{M^\sigma, N^\tau} \prl$ and 
$R \rabeta R' = M$, where \mbox{$\tr{\sigma} = \strzjn{\sigma} \ra
\bot$}.
%% and apply Rules (\ref{para}) and (\ref{rzut}).
%
 \begin{align*}
  \begin{split}
   \tr{R} & = \tr{\para{M, N}^{\sigma \wedge \tau} \prl} \\
          & = \lambda a_1^{\sigma_1} \ldots \lambda a_n^{\sigma_n} . 
          \tr{\para{M, N}}^\tr{\sigma \wedge \tau}
          (\lambda x^\tr{\sigma}. \lambda y^\tr{\tau}. (x a_1 \ldots a_n)^\bot) \\
          & = \lambda \vec{a} . 
          (\lambda z^{\tr{\sigma}\ra\tr{\tau}\ra\bot}.z \tr{M} \tr{N}) 
          (\lambda x^\tr{\sigma} \lambda y^\tr{\tau} . (x \vec{a})^\bot) \\
          & \rabeta \lambda \vec{a} . ((\lambda x^\tr{\sigma} \lambda y^\tr{\tau} . 
           (x \vec{a})^\bot) \tr{M} \tr{N}) \\
          & \rabeta \lambda \vec{a} . (\lambda y^\tr{\tau} . \tr{M} \vec{a})\tr{N} 
            \rabeta \lambda \vec{a} . \tr{M} \vec{a} \ra_\eta^{+} \tr{M} \\
          & = \tr{R'}
  \end{split}
 \end{align*}

 \noindent
(\ref{elimcase})
Let $R=(W^{\sigma\vee\tau}[x.S^{\alpha\ra\beta},y.T^{\alpha\ra\beta}]) 
N^\alpha$ and let 
$R' = W^{\sigma\vee\tau}[x.(SN)^\beta, y.(TN)^\beta]$. Then
$R \cmt R'$, according to 
%%reduction rule~
(\ref{elimcase}). 
Assuming $\tr{\beta} = \strzjn{\beta}\ra\bot$, we have
 \begin{align*}
  \begin{split}
   \tr{R} & = (\lambda a^\tr{\alpha} b_1^{\beta_1}\ldots b_n^{\beta_n}.\tr{W}
     (\lambda x^\tr\sigma.\tr{S}^{\tr\alpha\ra\tr\beta} a\vec{b})
     (\lambda y^\tr\tau. \tr{T}^{\tr\alpha\ra\tr\beta}a\vec{b})) \tr{N}^\tr{\alpha} \\
   & \rabeta \lambda b_1\ldots b_n.\tr{W}
     (\lambda x^\tr\sigma.\tr{S}\tr{N}\vec{b})
     (\lambda y^\tr\tau. \tr{T}\tr{N}\vec{b})) \\
   & = \tr{R'}
  \end{split}
 \end{align*}

\noindent  Other cases are similar.
\qed
\end{proof}

\begin{theorem} The calculus $\lambdaiakf$ is strongly normalizing.
\end{theorem}
\begin{proof}
 Suppose, by contradiction, that $M^\tau$ admits an
infinite $\beta$-reduction \[ M^\tau = M_0^\tau \rabeta M_1^\tau \rabeta M_2^\tau \rabeta \cdots \]
 By Theorem \ref{betaeta} we have an infinite reduction in $\lambda_\ra$ \[ \tr{M^\tau} = \tr{M_0} \thra_{\beta\eta}^{+} \tr{M_1} \thra_{\beta\eta}^{+} \tr{M_2} \thra_{\beta\eta}^{+} \cdots \] This contradicts the SN property of $\lambda_\ra$ 
 \qed
\end{proof}

\section{Translation for polymorphic types}\label{ksuFGUFFO}

%% The idea of CPS-translation from \cite{NakazawaTatsuta08} motivated development of following translation for the system $\Fiakf$. Particularly, rules for absurdity eliminator and universal and existential quantifiers does not appear in cited work. 
%% [Co autor chcial powiedziec?]

As we mentioned in the introduction, the translations in 
Section~\ref{ksuFGUFFO}
 %%Definitions \ref{trtyp} and \ref{trterm}
are not adequate for the polymorphic case and therefore we apply a call-by-name CPS translation. 
In general, a type $\tau$ is translated to $\trf{\tau} = (\tau^\ast\ra\bot)\ra\bot$. 
This translation, unlike 
the one for simple types, does not unify type constants. The helper translation ${}^\ast$ is given below.

\begin{definition}\rm Helper translation ${}^\ast$.
 \begin{align*}
  \alpha^\ast & = \alpha, \; \mbox{\textnormal {for all type constants} $\alpha = \bot, p, q, \ldots$} \\
  (\alpha \ra \beta)^\ast & = \trf{\alpha} \ra \trf{\beta} \\
  (\alpha \wedge \beta)^\ast & = (\trf{\alpha} \ra \trf{\beta} \ra \bot ) \ra \bot \\
  (\alpha \vee \beta)^\ast & = (\trf{\alpha} \ra \bot ) \ra (\trf{\beta} \ra \bot) \ra \bot \\
  (\forall p \, \tau)^\ast & = \forall p \, \trf{\tau} \\
  (\exists p \, \tau)^\ast & = (\forall p (\trf{\tau} \ra \bot))\ra\bot
 \end{align*}
\end{definition}

\noindent A %%
term $M^\tau$ is translated to the term $\trf{M} = \lambda k^{\tau^\ast\ra\bot}. (M \Ra k)$. 
To achieve that, two helper translations are needed: $\Ra$ and $\at$.
The term $K$ in the definition below is of type $\tau^\ast\ra\bot$.
The term $M \Ra K$ is always of type $\bot$.

\begin{definition}\rm Helper translation $\Ra$
 \begin{align}
  x^\tau \Ra K & = x K \\
  \lambda x^\sigma . N^\rho \Ra K & = K (\lambda x^\trf{\sigma} . \trf{N}) \\
  \para{N_1^{\tau_1}, N_2^{\tau_2}} \Ra K
   & = K (\lambda p\,^{\trf{\tau_1}\ra\trf{\tau_2}\ra\bot} . p \trf{N_1} \, \trf{N_2}) \\
  (\inl A)^{\tau_1 \vee \tau_2} \Ra K & = K (\lambda a\,^{\trf{\tau_1}\ra\bot} b\,^{\trf{\tau_2}\ra\bot} . a \trf{A}) \\
  (\inr B)^{\tau_1 \vee \tau_2} \Ra K & = K (\lambda a\,^{\trf{\tau_1}\ra\bot} b\,^{\trf{\tau_2}\ra\bot} . b \trf{B}) \\
  \mathrm\Lambda p \, N^\rho \Ra K & = K (\mathrm\Lambda p. \trf{N}) \\
  [N^{\rho[p:=\sigma]}, \sigma] \Ra K 
   & = K (\lambda u^{\forall p (\trf{\rho}\ra\bot)} . u \,\trf{\sigma} \, \trf{N}) \\
   NE \Ra K & = N \Ra (E \at K) \label{elim}
 \end{align}
 In (\ref{elim}) the symbol $E$ stands for an arbitrary eliminator. That is, 
 $E$ is one of the expressions
 $\lbrace R^\sigma, \prl, \prr, [x^{\tau_1}.S^\delta, y^{\tau_2}.T^\delta], \sigma, [x^\rho.S^\delta], \cud{\alpha}
 \rbrace$ and the omitted type of term $N$ is appropriate for every eliminator $E$.
\end{definition}

\begin{definition}\rm Helper translation $\at$
 \begin{align*}
  R \at K & = \lambda m^{\trf{\sigma} \ra \trf{\rho}} . m \trf{R} K \\
  \prl \at K & = \lambda m^{(\trf{\tau_1}\ra\trf{\tau_2}\ra\bot)\ra\bot} . m (\lambda a\,^\trf{\tau_1} \,b\,^\trf{\tau_2} . a K) \\
  \prr \at K & = \lambda m^{(\trf{\tau_1}\ra\trf{\tau_2}\ra\bot)\ra\bot} . m (\lambda a\,^\trf{\tau_1} \,b\,^\trf{\tau_2} . b K) \\
  \begin{split}
   [x^{\tau_1}.S^\delta, y^{\tau_2}.T^\delta] \at K & = \lambda m^{(\trf{\tau_1}\ra\bot)\ra(\trf{\tau_2}\ra\bot)\ra\bot} . \\ & \quad m (\lambda x\,^\trf{\tau_1}. (S\Ra K))(\lambda y\,^\trf{\tau_2}. (T\Ra K))
  \end{split} \\
  \sigma \at K & = \lambda m^{\forall p \trf{\rho}} . m \trf{\sigma} K \\
  [x^\rho.S^\delta] \at K & = \lambda m^{(\forall p(\trf{\rho}\ra\bot))\ra\bot} . m (\mathrm\Lambda p \lambda x\,^\trf{\rho}. (S \Ra K)) \\
  \cud{\alpha} \at K & = \lambda m^\bot . m
 \end{align*}
\end{definition}

\begin{lemma} \emph{[Soundness]}
\label{trsoundf} If a term $M$ has type $\delta$, then $\trf{M}$ has type $\trf{\delta}$.
\end{lemma}

\begin{proof}
Easy.
\qed
\end{proof}

\begin{lemma} \emph{[Properties of substitution]} For a term $R$ and 
any term $K$ and for any types $\tau$ and $\rho$ the 
following holds:
\label{substr}
 \begin{align}
  \trf{R}[x^\trf \delta:=\trf N^ \trf \delta] &=_{\alpha} \trf{R[x:=N]}; \\
  (R \Ra K)[x ^ \trf\delta:=\trf N ^ \trf\delta] &=_{\alpha} R [x := N] \Ra K [x := N]; \label{podstterm} \\
  (R \at K)[x ^ \trf\delta := \trf N ^ \trf\delta] &=_{\alpha} R[x := N] \at K[x := N]
   \mbox{\textnormal{ if $R$ is an eliminator}}; \\
  \trf\tau [p:=\trf\rho] &=_\alpha \trf{\tau[p:=\rho]}; \\
  (R \Ra K)[p:=\trf\rho] &=_\alpha R[p:=\rho] \Ra K[p:=\rho]; \label{podsttyp}\\
  (R \at K)[p:=\trf\rho] &=_\alpha R[p:=\rho] \at K[p:=\rho] \mbox{\textnormal{ if $R$ is an eliminator}}.
 \end{align}
\end{lemma}

\begin{proof} This lemma is proved by simultaneous induction on the definition of substitution. 
 \qed
\end{proof}

\begin{lemma} \label{rbetaf} If $R \rabeta R'$, then $\trf{R} \rabeta^{+} \trf{R'}$.
\end{lemma}

\begin{proof}
 Using induction on the definition of $\rabeta$ we have 7 cases. For example, consider (\ref{elimegz}),
where $R=[M^{\tau[p := \sigma]}, \sigma][x^\tau.N^\delta]$ and $R' = (N[p:=\sigma][x := M])^\delta$.
 \begin{align*}
 (\ref{elimegz}) \quad \trf R & = \lambda k . (\lambda m^{(\exists p \tau)^\ast} . m
    (\mathrm\Lambda p \lambda x^\tau. (N \Ra k))) (\lambda u^{\forall p (\trf\tau\ra\bot)} . u \trf \sigma\, \trf M) \\
  & \rabeta \lambda k . (\lambda u . u \trf \sigma \, \trf M) (\mathrm\Lambda p \lambda x. (N \Ra k)) \\
  & \rabeta \lambda k . (\mathrm\Lambda p \lambda x. (N \Ra k)) \trf \sigma \,\trf M \\
  & \rabeta \lambda k . (\lambda x. (N \Ra k))[p:=\trf\sigma] \,\trf M \\
  & \rabeta \lambda k . (\lambda x. (N [p:=\sigma]\Ra k)) \trf M \mbox{\quad (from (\ref{podsttyp}))} \\
  & \rabeta \lambda k . (N [p:=\sigma]\Ra k) [x:=\trf M] \\
  & =_\alpha \lambda k . (N [p:=\sigma][x:=M] \Ra k) \quad \mbox{(from (\ref{podstterm}))} \\
  & = \trf {R'}
 \end{align*}
\qed
\end{proof}

\begin{lemma} \label{rcmtf} If $R \cmt R'$, then $\trf R =_\alpha \trf R'$.
\end{lemma}

\begin{proof}
The complete proof consists of 21 cases. Here, two interesting commutations will be elaborated. 
 The other cases are similar and left to the reader.

 From (\ref{Eelimelim}) we get
 \begin{align*}
  \begin{split}
   \trf{\lhs} = & \lambda k . (M[y.P] \Ra ([x.N] \at k)) = \lambda k. (M \Ra ([y.P] \at ([x.N] \at k))) \\
    = & \lambda k . (M \Ra (\lambda m .m (\mathrm\Lambda p\lambda y.(P \Ra [x.N]\at k))))
  \end{split} \\
  \begin{split}
   \trf{\rhs} = & \lambda k. (M \Ra ([y. P[x.N]] \at k)) = \lambda k. (M \Ra (\lambda m.m(\mathrm\Lambda p \lambda y. (P[x.N] \Ra k)))) \\
    = & \lambda k . (M \Ra (\lambda m .m (\mathrm\Lambda p\lambda y.(P \Ra [x.N]\at k))))
  \end{split}
 \end{align*}
 
 \noindent From (\ref{Aelimcud}) we get
 \begin{align*}
  \begin{split}
   \trf{\lhs} = & \lambda k . (A[x.N] \Ra (\cud{\sigma} \at k)) = \lambda k. (A[x.N] \Ra (\cud{\sigma} \at k)) \\
    = & \lambda k. (A \Ra ([x.N] \at (\cud{\sigma} \at k))) \\
    = & \lambda k . (A \Ra (\lambda m . m (\mathrm\Lambda p \lambda x . (N \Ra (\cud{\sigma} \at k)))))
  \end{split} \\
  \begin{split}
   \trf{\rhs} = & \lambda k . (A \Ra ([x.N\cud{\sigma}] \at k)) = \lambda k. (A \Ra (\lambda m . m(\mathrm\Lambda p \lambda x.(N \cud{\sigma} \Ra k)))) \\
   = & \lambda k . (A \Ra (\lambda m . m (\mathrm\Lambda p \lambda x . (N \Ra (\cud{\sigma} \at k)))))
  \end{split}
 \end{align*}
\qed
\end{proof}

\begin{lemma} \label{cmtsteps} Every sequence of commutative reductions
in $\Fiakf$ must terminate.
\end{lemma}
\begin{proof}
 To prove this lemma we define such a measure $\chi(M) > 0$, that for any commutation $M\cmt M'$, we have $\chi(M) > \chi(M')$. Please note, that we have 3 patterns of commutative reductions in Rules (\ref{cmts-pat1}), (\ref{cmts-pat2}) and (\ref{cmts-pat3}). We use those patters to define appropriate conditions for measure $\chi$:
 \begin{align}
   \chi\left((W[x.S,y.T])E\right) &> \chi\left(W[x.SE,y.TE]\right) \label{chipat1} \\
   \chi\left((A\cud{{}})E\right) &> \chi\left(A\cud{{}}\right) \label{chipat2} \\
   \chi\left((N[x.P])E \right) &> \chi\left(N[x.PE]\right) \label{chipat3} \\
   \chi(M) &\geq 1 \nonumber
 \end{align}
 Now we give the definition
of the function $\chi(M)$;
 it is similar to de Groote's norm $| \cdot |$ from \cite{groote99} but simpler:
 \begin{align*}
  & \chi(x) = 1 \\
  & \chi(\lambda x . N) = \chi (\inl N) = \chi (\inr N) = \chi (N), \quad \chi(\para{M_1,M_2}) = \chi(M_1) + \chi (M_2) \\
  & \chi(F A) = \chi(F)^2 \chi(A), \quad \chi(P \prl) = \chi (P \prr) = \chi(P)^2, \quad \chi(N \sigma) = \chi(N)^2 \\
  & \chi(W[x.S,y.T]) = \chi(W)^2(\chi(S)+\chi(T)) + 1 \quad \chi(N[x.P]) = \chi(N)^2\chi(P) + 1 \\
  & \chi(A \cud{{}}) = \chi(A)^2 + 1 \\
 \end{align*}
 \noindent 
There are 21 easy cases, one for each permutation from Definitions \ref{cmts} and \ref{cmtf}. We will show here one example case for each pattern mentioned above. \medskip

\noindent
 (\ref{chipat1}) Let $l=\chi((W[x.S,y.T])[a.A,b.B])$ and $r=\chi(W[x.S[a.A,b.B],y.T[a.A,b.B]])$.
{\allowdisplaybreaks
 \begin{align*}
  l & =\chi(W[x.S,y.T])^2(\chi(A)+\chi(B)) + 1 \\
    & = \left( \chi(W)^2(\chi(S)+\chi(T)) + 1 \right)^2(\chi(A)+\chi(B)) + 1 \\
    & > \left( \chi(W)^2(\chi(S)+\chi(T)) \right)^2(\chi(A)+\chi(B)) + 1 \\
    & = \chi(W)^4\left((\chi(S)^2+\chi(T)^2) (\chi(A)+\chi(B)) + 2(\chi(S)\chi(T))(\chi(A)+\chi(B))\right) + 1\\
    & > \chi(W)^4((\chi(S)^2+\chi(T)^2) (\chi(A)+\chi(B)) + 2) + 1\\
  r & = \chi(W)^2 (\chi(S[a.A,b.B]) + \chi(T[a.A,b.B])) + 1\\
    & = \chi(W)^2 (\chi(S)^2(\chi(A)+\chi(B)) + 1 + \chi(T)^2(\chi(A)+\chi(B)) + 1) + 1 \\
    & = \chi(W)^2 ((\chi(S)^2+\chi(T)^2)(\chi(A)+\chi(B)) + 2) + 1 \\
   l & > r
 \end{align*}
}
 
 \noindent
(\ref{chipat2}) Let $l=\chi((A\cud{\bot})\cud{\sigma})$ and $r=\chi(A\cud{\sigma})$.
 \begin{align*}
  l & = \chi(A\cud{\bot})^2 + 1 = (\chi(A)^2 + 1)^2 + 1 = \chi(A)^4 + 2\chi(A)^2 + 2 \\
  r & = \chi(A)^2 + 1 \\
  l & > r
 \end{align*}
 \noindent
(\ref{chipat3}) Let $l=\chi((N[x.P])[a.A,b.B])$ and $r=\chi(N[x.P[a.A,b.B]])$.
{\allowdisplaybreaks
 \begin{align*}
  l & = \chi(N[x.P])^2 (\chi(A) + \chi(B)) + 1 \\
    & = \left( \chi(N)^2\chi(P) + 1 \right)^2 (\chi(A) + \chi(B)) + 1 \\
    & = (\chi(N)^4\chi(P)^2 + 2\chi(N)^2\chi(P) + 1)(\chi(A) + \chi(B)) + 1 \\
    & = \chi(N)^4\chi(P)^2(\chi(A) + \chi(B)) + \chi(N)^2(2\chi(P)(\chi(A) + \chi(B))) \\
    & \quad + \chi(A) + \chi(B) + 1 \\
  r & = \chi(N)^2\chi(P[a.A,b.B]) + 1 \\
    & = \chi(N)^2(\chi(P)^2(\chi(A) + \chi(B)) + 1) + 1 \\
    & = \chi(N)^2\chi(P)^2(\chi(A) + \chi(B)) + \chi(N)^2+ 1 \\
  l & > r
 \end{align*}
}
 \qed
\end{proof}

\begin{theorem} The calculus $\Fiakf$  is strongly normalizing.
\end{theorem}
\begin{proof}
 Suppose that 
 \[ M^\tau = M_0^\tau \ra M_1^\tau \ra M_2^\tau \ra \cdots \]
If there is infinitely many $\beta$-reductions in the sequence above then
we have an infinite reduction in $\Fi$. If almost all reduction steps
are of type $\cmt$ then we use Lemma \ref{cmtsteps}. In both cases we reach contradiction.
 \qed
\end{proof}

\section{Summary}

We have presented a short proofs of strong normalization for simply-typed and polymorphic $\lambda$-calculus with all connectives. Syntax-driven translations used in those proofs allow to reduce the SN property problem to calculi with less number of connectives. 

The CPS-translation used for polymorphic case looks may be helpful dealing with higher level $\lambda$-calculus such as $\mathbf F_\omega$. This is our next research problem.

\bibliographystyle{plain}
\bibliography{paper}

\end{document}